\documentclass[10pt]{IEEEtran}
\IEEEoverridecommandlockouts
\usepackage{url}
\usepackage{verbatim}
\usepackage{cite}
\usepackage{amsmath,amssymb,amsfonts}
\usepackage{lipsum}
\usepackage{caption}
\usepackage{graphicx}
\usepackage{textcomp}
\usepackage{xcolor}
\usepackage{subfigure}
\usepackage{amsthm}
\usepackage{dsfont}
\usepackage{algorithmic}
\usepackage{algorithm}
\newtheorem{definition}{Definition} [section]
\newtheorem{theorem}{Theorem}  [section]

\ifCLASSINFOpdf

\else

\fi

\begin{document}
\title{Blockchain-based  Edge Resource Sharing for Metaverse
}

\author{
\thanks{This work is partly supported by the US NSF under grant CNS-2105004.}
\IEEEauthorblockN{Zhilin Wang$^{\dag}$, Qin Hu$^{\dag}$ (Corresponding author), Minghui Xu$^{\ddag}$, Honglu Jiang$^\S$}\\
\IEEEauthorblockA{$^\dagger$Department of Computer \& Information Science, 
Indiana University-Purdue University Indianapolis,  USA \\}
\IEEEauthorblockA{$^{\ddag}$School of Computer Science \& Technology, Shandong University, China\\}
\IEEEauthorblockA{$^\S$Department of Computer Science \& Software Engineering, Miami University,  USA \\}
Email: \{wangzhil, qinhu\}@iu.edu, mhxu@sdu.edu.cn, jiangh34@miamioh.edu
}

\maketitle

\begin{abstract}
Although Metaverse has recently been widely studied, its practical application still faces many challenges. One of the severe challenges is the lack of sufficient resources for computing and communication on local devices, resulting in the inability to access the Metaverse services. To address this issue, this paper proposes a practical blockchain-based mobile edge computing (MEC) platform for resource sharing and optimal utilization to complete the requested offloading tasks, given the heterogeneity of servers' available resources and that of users' task requests. 
To be specific, we first elaborate the design of our proposed system and then dive into the task allocation mechanism to assign offloading tasks to proper servers. To solve the multiple task allocation (MTA) problem in polynomial time, we devise a learning-based algorithm. Since the objective function and constraints of MTA are significantly affected by the servers uploading the tasks, we reformulate it as a reinforcement learning problem and calculate the rewards for each state and action considering the influences of servers.
Finally, numerous experiments are conducted to demonstrate the effectiveness and efficiency of our proposed system and algorithms.

\end{abstract}


\section{Introduction}
Metaverse, which is considered as the next generation of the Internet, has attracted researchers' attention recently~\cite{xu2022full, mystakidis2022metaverse}. In Metaverse, people can interact with the virtual world through technologies like virtual reality and augmented reality. Currently, people typically access the servers of Metaverse providers through Metaverse local devices, such as head-mounted glasses, to get the Metaverse services. These devices are required to assist users in accessing and experiencing Metaverse services, as well as to process resource-intensive computing tasks locally. 
However, using Metaverse local devices for computing faces severe challenges: 1) the computing and communication resources of the devices are limited; 2) the locations of devices performing local computing are dispersed and the devices may be constantly moving.



Fortunately, there is one existing solution for addressing these challenges, namely mobile edge computing (MEC)~\cite{abbas2017mobile, huda2022survey}. Specifically, local devices can offload their computing tasks to proxy MEC servers, e.g., base stations; then, the MEC servers finish those tasks and return the results to local devices. Since the MEC servers are located close to local devices and usually have enough computing resources, their involvement can reduce the latency of communication and computing, thus providing low-latency and high-quality Metaverse services.  


In practice, the MEC servers in Metaverse are usually responsible for computing multiple tasks from different users. Since their resources are not infinite, one single MEC server may not be able to handle those offloading tasks in time, leading to low quality of Metaverse services. One possible solution is to involve other MEC servers with extra resources to work together for offloading computing in Metaverse, which makes it necessary to establish a secure resource trading platform 
for edge servers. To that aim, we propose a blockchain system running on MEC servers to form a distributed computing framework, named blockchain-based MEC platform, 
which enables resource integration and optimal utilization among MEC servers in a trustless environment.

Currently, there is no existing study focusing on implementing blockchain-based MEC for resource sharing and optimization in Metaverse. 
Although there exist several studies about MEC in Metaverse, they focus on the latency analysis~\cite{dhelim2022edge} and incentive mechanism design~\cite{xu2021wireless}. While for research about blockchain-based MEC~\cite{sheng2020near, zhaofeng2019blockchain}, no one has considered solving the resource sharing problem in MEC.

To fill this gap, the proposed blockchain-based MEC platform aims to assist resource sharing and optimization in Metaverse via trading offloading tasks in a transparent but secure way, so that more offloading tasks from Metaverse users can be finished timely. This platform comprises multiple Metaverse users, MEC servers, and a consortium blockchain system running the practical Byzantine fault tolerance (PBFT) consensus protocol~\cite{castro1999practical, xu2021concurrent}. And there are four main procedures in the proposed system, i.e., \textit{data submission, task allocation, offloading computing}, and \textit{payment of tasks}. As the pivotal step, task allocation faces a critical challenge brought by the heterogeneity of multi-task requests from users. Specifically, given the limited computing and communication resources of each MEC server and offloading tasks with different price policies, data sizes, and completion time requirements, our proposed system needs to assign the requested multiple tasks to multiple servers under various constraints. In addition, since the offloading tasks are time-sensitive,  the task allocation step is expected to make the (near) optimal decisions as fast as possible so as to reduce the latency of the whole system.


To address these challenges, we design a learning-based task allocation mechanism. Specifically, we first formulate the task allocation issue as an integer-programming problem; then, we analyze the utilities of different decisions by considering the resource constraints of servers and the time requirements of tasks; according to our analysis, the multiple tasks allocation (MTA) problem is NP-complete, and the time complexity is too high, so we design a learning-based algorithm to find its approximate optimal value in polynomial time. 
Since the objective function and constraints of the MTA problem are conditional, which are affected by the servers uploading the tasks to blockchain, in our proposed learning-based solution, we first transform the MTA problem into a reinforcement learning problem and then calculate the rewards for each state and action  considering the influences of the server source of tasks. 

To the best of our knowledge, we are the first to implement blockchain-based MEC in Metaverse. The main contributions  of this paper are summarized as below:
\begin{itemize}
    \item We are the first to propose a practical blockchain-based MEC platform for resource sharing and optimal utilization in Metaverse. Our proposed system can satisfy the demands of low latency, multiple requests, energy efficiency, incentive compatibility, and data privacy protection for Metaverse users.
    \item We design a task allocation mechanism for edge resource sharing. Under the heterogeneous constraints, we first formulate an MTA problem for reasonably distributing offloading tasks among multiple MEC servers.
    \item We propose a learning-based solution to find the approximate optimal solution for the MTA problem with polynomial time complexity. We speed up the learning process via identifying available actions 
    for each state.
    \item We conduct extensive experiments to verify the effectiveness, efficiency, and validity of our proposed system, mechanisms, and algorithms. 
\end{itemize}


\label{related}

\section{System Model}
\label{system}


\subsection{System Overview}

Our proposed blockchain-based edge resource sharing platform is illustrated in Fig.~\ref{fig_sys} for supporting the Metaverse applications, consisting of mobile devices as Metaverse users, MEC servers, and the consortium blockchain running with practical Byzantine fault tolerance (PBFT) consensus. Here we assume that users in our considered system are devices with offloading requests to MEC servers. Based on the arrival time of offloading tasks in the blockchain network, we assign a specific task number to each of them. 
Specifically, we define $\mathcal{T}=\{t_1,\cdots,t_j,\cdots,t_m\}$ as the set of offloading tasks from users 
with $m$ denoting the number of all offloading tasks. 
Let $\mathcal{S}=\{s_1,\cdots,s_i,\cdots,s_n\}$ denote the set of MEC servers in the Metaverse
with $n$ being the total number of servers. 

\begin{figure}[h!t]
\centering
\includegraphics[width=\linewidth]{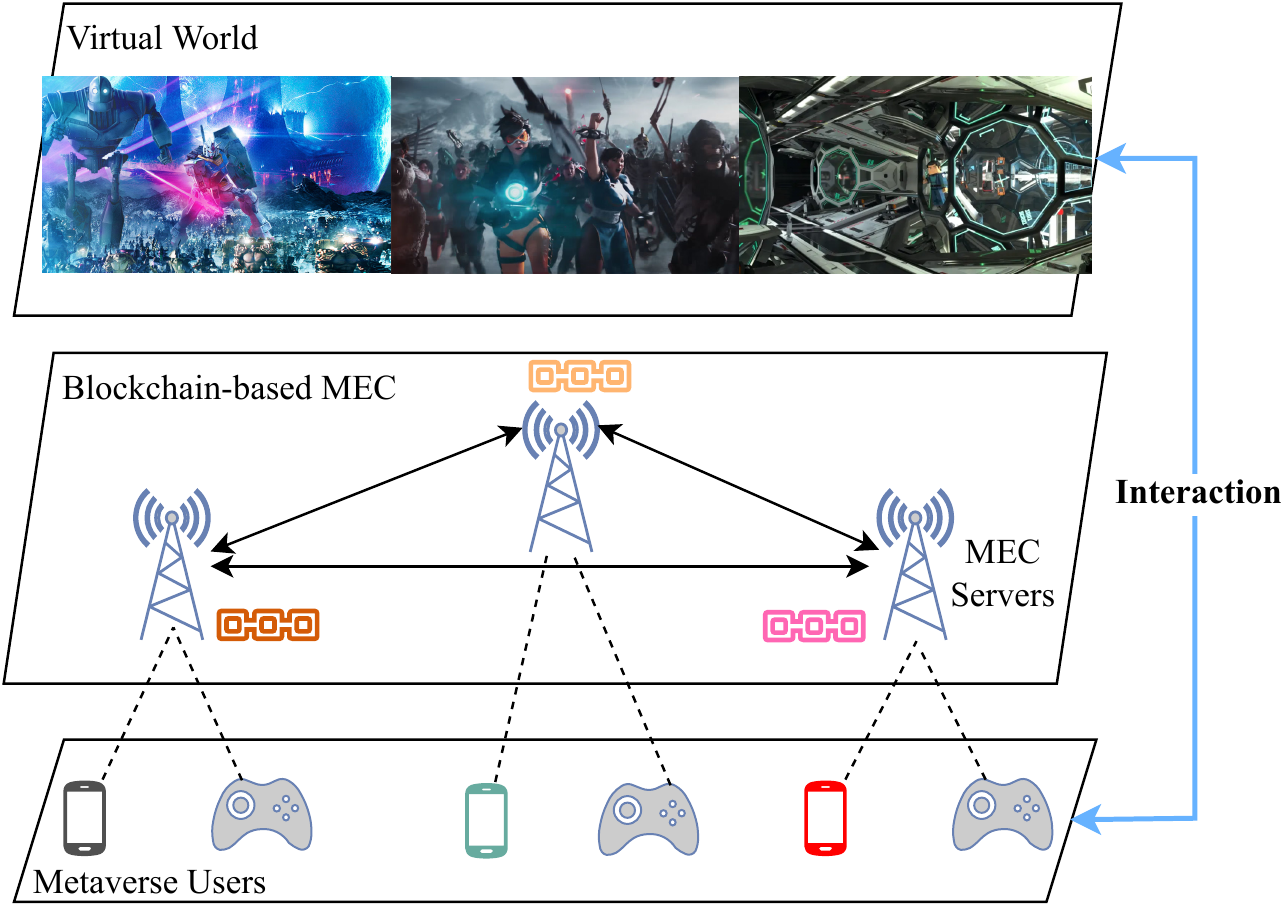}
\caption{The illustration of our proposed system.} 
\label{fig_sys}
\end{figure}

The workflow of our proposed system can be described as below: 1) \textit{Data Submission}: The Metaverse users upload their raw data and the descriptions of tasks to their nearest MEC servers. We use $R_j<p_j, D_j, \tau_{e,j}>$ to denote the description of task $t_j$, 
where $p_j$ is the unit price of per CPU cycle of computing for finishing $t_j$,  $D_j$ is the data size of $t_j$, and $\tau_{e,j}$ is the time requirement for finishing $t_j$. 2) \textit{Task Allocation}: Once server $s_i$ receives the data and $R_j$ from users, the data is stored locally and $R_j$ is submitted to the blockchain for task allocation which will be elaborated in the following sections. 3) \textit{Offloading Computing}: According to the task allocation results, server $s_i$ will either send the data of tasks it received to appropriate servers or finish the computing of the tasks locally. Then, the servers start to work on the allocated tasks and will broadcast the computing results to the blockchain network when the tasks are completed. 4) \textit{Payment of Tasks}: The users first pay offloading computing fees through the blockchain network where they can also get the computing results later. Each server can get its own payment through the coordination of the blockchain. 

To protect the privacy of the computing results, we can employ the asymmetric encryption technique, such as Rivest–Shamir–Adleman (RSA) algorithm~\cite{izdemir2021rivest}, in this system. Specifically, at the beginning of each round, the users submit the public keys generated by RSA along with the offloading task requests while keep their private keys secret. After the tasks are finished by the MEC servers, the computing results will be encrypted with the public keys of the corresponding users. Then, the users can download the encrypted results via accessing the blockchain
and decrypt the results with their private keys.
Since the amount of tasks received by each server can be different, we define the tasks submitted by server $s_i$ as $\mathcal{T}_{s,i}=\{t_1,t_2,\cdots,t_{m'}\}$, where $m'\leq m$ is the number of tasks from $s_i$. For example, if  $t_1$, $t_3$, and $t_{20}$ are submitted by $s_1$, then we have $\mathcal{T}_{s, 1}=\left\{t_1,t_3,t_{{20}}\right\}$.

\subsection{Consortium Blockchain}

The consortium blockchain is running on the MEC servers, 
where all servers are authorized nodes and thus can be trusted. Besides, as mentioned earlier, we implement PBFT, a lightweight consensus protocol in blockchain, to assist servers in reaching consensus for generating blocks to help disseminate task descriptions and allocation results, as well as enforce payment distribution 
in a round-by-round manner. 

Here is the workflow of the blockchain network in our system: 1) At the beginning of each round, each server publishes its available resources for finishing tasks, including both computing and communication resources (see Section~\ref{model} for details), and uploads received task descriptions $R_j$ to the blockchain network. 2) All of the above-mentioned information will be broadcast in the blockchain network so that the leader node can allocate the received tasks to appropriate servers with the help of the task allocation mechanism detailed in the next section. Then the allocation decisions, resource information of the MEC servers, and task descriptions will be packaged into a new block. 3) If $s_i$ is assigned to process the tasks submitted by itself, it can start processing right away; otherwise, it has to transfer the received data to other servers for finishing the offloading tasks. 4) After tasks are finished, the results will be broadcast so as to be recorded on the blockchain. 5) The users can get the results from the blockchain, and the servers will be paid accordingly.



Overall, the blockchain offers a decentralized and trusted platform to conduct resource sharing for MEC empowered Metaverse applications, thus enabling offloading computing for Metaverse users to overcome the challenges posed by the insufficient computing power of individual server and device mobility. With the combination of MEC and blockchain, our proposed resource sharing platform can handle multiple offloading tasks in a efficient and distributed manner. However, regarding the pivotal step, we need to design a task allocation mechanism to assign offloading tasks to proper MEC servers so as to achieve the optimal resource utilization while satisfying user requests, which will be discussed in the next section.

\label{mechanism}

\section{Design of Task Allocation Mechanism }
\label{model}
In this section, we detail the design of task allocation mechanism by modeling all possible computing and communication costs for $s_i$

\subsection{Computing Cost Model}
According to~\cite{wang2022incentive}, if task $t_j$ is allocated to server $s_i$, 
its energy cost can be calculated as $E_{i,j}^\text{comp}=\alpha_i \mu_{i,j} f_{i}^2$,
where $\alpha_i$ is the parameter related to the architecture of CPU, $\mu_{i,j}=D_j \theta_{i}$ is the total CPU cycles required to finish task $t_j$ on the MEC server $s_i$ with $\theta_{i}$ being the CPU cycles required to process unit data sample, and $f_i$ is the CPU frequency used to compute the offloading tasks. Besides, the time consumption of computing $t_j$ on $s_i$ can be calculated by $T_{i,j}^\text{comp}=\frac{\mu_{i,j}}{f_{i}}$.
Let $\mu_i$ be the maximum available computing capacity of $s_i$. Then $\alpha_i, \theta_i, f_i$ and $\mu_i$ will be submitted to the blockchain as the available computing resource of $s_i$ for task allocation.

\subsection{Communication Cost Model} 
If $t_j$ is submitted by $s_i$ but not finally assigned to $s_i$, $s_i$ would only be responsible for transmitting the data of $t_j$ to the determined server.
We define $B_{i}$ as the allocated bandwidth of $s_i$ for data transmission; let $H_{i}$ and $G_{i}$ be the transmission power and channel power gain, respectively.
Based on Shannon Bound, we can get the data transmission rate as $r_{i}=B_{i}\log_2 (1+\frac{H_{i} G_{i}}{\delta^2})$,
where $0\leq \delta \leq 1$ is the Gaussian noise during the transmission. Then we can calculate the transmission time as $T_{i,j}^\text{comm}=\frac{D_j}{r_{i}}$.
And the energy consumption of data transmission can be calculated by $E_{i,j}^\text{comm}=H_{i} T_{i,j}^\text{comm}$.
Note that $B_i, H_i$, and $G_i$ are uploaded to the blockchain network as the available communication resource of $s_i$.


\subsection{Utility Model} 
After $t_j$ is completed, server $s_i$ will receive the payment from users, denoted by $p_j\mu_{i,j}$ for task $t_j$, which is the product of unit price per CPU cycle and the total number of consumed CPU cycles. Then, the utility of $s_i$ regarding finishing $t_j$ in the system is the difference between the received payment and the local cost. 

Since there are two possible task allocation results for any server $s_i$ with respect to processing task $t_j$, i.e., computing $t_j$ locally or not, which will significantly affect the utility of $s_i$, we define an indicator $\mathds{1}_{i,j}$ to capture this:  if $t_j$ is assigned to $s_i$, then $\mathds{1}_{i,j}=1$; otherwise, $\mathds{1}_{i,j}=0$.
If $t_j \in \mathcal{T}_{s,i}$, the utility can be expressed as $U_{i,j}^{I}= (p_j \mu_{i,j}-E_{i,j}^\text{comp}) \mathds{1}_{i,j}+(\lambda p_j \mu_{i,j}-E_{i,j}^\text{comm})(1-\mathds{1}_{i,j})$, 
where $\lambda p_j \mu_{i,j}$ is the intermediary fee paid by the server finishing $t_j$, with $\lambda$ being a predefined constant parameter that can be known globally;
and the total time consumption is $T_{i,j}^I=T_{i,j}^\text{comp}\mathds{1}_{i,j}+T_{i,j}^\text{comm} (1-\mathds{1}_{i,j})$.  While if $t_j \notin \mathcal{T}_{s,i}$, the utility is defined as $U_{i,j}^N=[(1-\lambda)p_j \mu_{i,j}-E_{i,j}^\text{comp}] \mathds{1}_{i,j}$, and the total time consumption is $T_{i,j}^N=(T_{i,j}^\text{comp}+T_{i^*,j}^\text{comm})\mathds{1}_{i,j}$, where $T_{i^*,j}^\text{comm}$ is the transmission time of server submitting $t_j$.
Thus, we can get the utility function and time consumption as: if $t_j \in \mathcal{T}_{s,i}$, $U_{i,j}=U_{i,j}^I$ and $T_{i,j}=T_{i,j}^I$; otherwise, $U_{i,j}=U_{i,j}^N$ and $T_{i,j}=T_{i,j}^N$.




\subsection{Problem Formulation}
\label{problem_formulation}
Recall the goal of task allocation for achieving resource integration and optimal utilization in our system, we would like to make sure that more tasks can be proceed in time and thus servers can get more rewards. So we can formulate it into a multi-task allocation (MTA) problem as follows:
\begin{align*}
\mathbf{MTA:} &\mathop{\arg\max}_{\mathds{1}_{i,j}} \sum_{i=1}^n \sum_{j=1}^m U_{i,j} \\
\text{s.t.}:
\text{C1}&:\sum_{j=1}^{m} \mu_{i,j} \mathds{1}_{i,j} \leq \mu_i,  \\
\text{C2}&: T_{i,j}\leq \tau_{e,j},   \\
\text{C3}&:\sum_{j=1}^m \mathds{1}_{i,j}\leq 1,  \\
\text{C4}&: i\in\left\{1,2,\cdots,n\right\},j\in\left\{1,2,\cdots,m\right\}, 0< n \leq m. 
\end{align*}
In the above MTA problem, the optimization objective is to maximize the total utility of all servers given requested tasks from Metaverse users; 
C1 is the computing capacity constraint to make sure that the selected tasks can be proceeded by server $s_i$; C2 is the time constraint, consisting of computing time and transmission time to ensure that $t_j$ can be completed in time; C3  guarantees that one task can only be assigned to one server, and it is possible that some servers are not assigned with tasks; C4 defines the domain of this optimization problem.

\begin{theorem}
MTA is an NP-complete problem.
\end{theorem}

\begin{proof}
The MTA problem can be described as a complex multiple knapsacks problem (MKP), which has been proved to be NP-complete in~\cite{li2019multi}. Specifically, MKP is defined as: let $\mathcal{T'}=\left\{t_1,t_2,\cdots,t_{m'}\right\}$ denote the set of $m$ items, and let $\mathcal{S'}=\left\{s_1,s_2,\cdots,s_{n'}\right\}$ as the set of $n$ knapsacks, and $m'>n'$. Each item $t_i'$ has its weight $D_j'$ and value $p_j'$ and each knapsack has its maximum weight $D_i'$. The goal is to decide where should each item be placed in the knapsacks so that the total utility is maximized. In other words, it aims to solve the following problem: $\max_{\mathds{1}_{i,j}}\sum_j^m \sum_i^n p_j'\mathds{1}_{i,j} $, where $\mathds{1}_{i,j}$ is an indicator: when $\mathds{1}_{i,j} = 1$, item $t_j'$ will be assigned to knapsack $s_i'$, and otherwise, $\mathds{1}_{i,j} = 0$. 
The time complexity is $O(n^{m})$, which means MKP cannot be solved in polynomial time. MKP can be reduced to the simplified MTA problem. 
In the original MTA problem, we need to allocate multiple tasks to multiple servers with computing resource and time consumption constraints. The objective of MTA is to maximize the total utility, i.e.,  $\max_{\mathds{1}_{i,j}}\sum_j^n \sum_i^m U_{i,j}$. If we remove the time constraints and assume that every server offers the same computing and communication resources for each task, we get the simplified MTA problem which is equivalent to MKP. As we can see, the simplified MTA also has the time complexity of $O(n^{m})$, so MTA is an NP-complete problem.
\end{proof}



\section{Learning-based Solution for MTA}
\label{solution}
Based on the analysis in Section~\ref{problem_formulation}, we know that the MTA problem is NP-complete, so we need to design a computationally efficient algorithm to solve it. In this section, we design a learning-based algorithm with $\epsilon$-greedy strategy.

\subsection{Problem Reformulation}
To begin with, we need to reformulate the MTA problem based on the Q-learning algorithm~\cite{jiang2020q, chien2020q}. There are three main components in Q-learning, i.e., state space, action space, and reward function, which are detailed as follows.
\subsubsection{State Space} 

We denote the state space of the agent, i.e., the MEC server executing the Q-learning algorithm to allocate tasks, as $\Omega <\mathcal{T},\mathcal{P}, \mathcal{D}, \mathcal{T}_{e} >$. Specifically, $\mathcal{T}=\left\{t_1,t_2,\cdots,t_m\right\}$ is the set of tasks, 
$\mathcal{P}$ is the set of the unit prices of tasks, i.e., $\mathcal{P}=\left\{p_1,p_2,\cdots,p_m\right\}$,  $\mathcal{D}=\left\{D_1,D_2,\cdots,D_m\right\}$ is the set of data sizes of tasks, and the set of execution time is defined as $\mathcal{T}_{e}=\left\{\tau_{e,1},\tau_{e,2},\cdots,\tau_{e,m}\right\}$. In other words, the state space is composed of all tasks with the corresponding descriptions, including prices, data sizes, and execution time requirements. Thus, $\Omega$ is a matrix with $m$ rows and 4 columns. The agent selects an action for each state based on the current task requirements and resource conditions of all servers. Once the action is chosen at a state, the agent will turn to the next state to conduct action selection.

\subsubsection{Action Space} 
Since servers have different amount of available resources, such as total CPU cycles, CPU cycle frequencies, and communication bandwidth, we can denote the action space as $\mathcal{A}<\mathcal{S},\mathcal{M},\mathcal{F},\mathcal{B},\mathcal{H},\mathcal{G},\alpha,\Theta >$. In detail, $\mathcal{S}$ is the set of servers, and $\mathcal{M}=\left\{\mu_1,\mu_2,\cdots,\mu_n\right\}$ is the set of total available CPU cycles, and $\mathcal{F}=\left\{f_1,f_2,\cdots,f_n\right\}$ is the set of CPU frequencies, and $\mathcal{B}=\left\{B_1,B_2,\cdots,B_n\right\}$ is the set of communication bandwidth; as for $\mathcal{H}=\left\{H_1,H_2,\cdots,H_n\right\}$ and $\mathcal{G}=\left\{G_1,G_2,\cdots,G_n\right\}$, they are the sets of transmission power and channel power gain, respectively; and $\alpha=\left\{\alpha_1,\alpha_2,\cdots,\alpha_n\right\}$ is the set of the parameter correlated to the CPU architectures and $\Theta=\left\{\theta_1,\theta_2,\cdots,\theta_n\right\}$ is the set of CPU cycles required for processing one data sample.
Specifically, there are two statuses for each action, i.e., \textit{selected} and \textit{not selected}, and the agent can choose only one action in one state while one action can be selected multiple times in all the states. This is to ensure that one task can only be assigned to one server; however, one server could process multiple tasks if it has sufficient resources. In this way, we can know that the action space is an $n\times8$ matrix. 

\subsubsection{Reward Function}
\label{reward_function}
The objective of MTA problem is to maximize the utility by allocating tasks to proper servers, so the rewards here are defined by the utilities.
In other words, the rewards are determined by the allocation decision, resource conditions, and time constraints. Besides, the rewards would also be affected by where the task is submitted from to the blockchain network. Thus, in the design of the reward function, we need to consider all of these aspects. 

We evaluate whether our system can process $t_j$ by two criteria: 1) server $s_i$ has enough computing power to be devoted to the computation, i.e., $\mu_{i,j}\leq\mu_i$; and 2) $s_i$ is able to process the task $t_j$ within the required time constraint, i.e., $T_{i,j}\leq \tau_{e,j}$. The second criteria is C2 while the first criteria is a weak C1. And C1 can only be used when selecting actions and updating Q-value, which will be discussed in Section~\ref{learning process}.


If $t_j\in T_{s_i}$, the reward function can be expressed as:
\begin{align}
U_{i,j}=
\begin{cases}
p_j \mu_{i,j}-E_{i,j}^\text{comp},& \mu_{i,j}\leq\mu_i\ \text{and}\ T_{i,j} \leq \tau_{e,j},\\
\lambda p_j \mu_{i,j}-E_{i,j}^\text{comm},& \text{otherwise}.
\end{cases}
\label{reward_1}
\end{align}

When $t_j$ is from $s_i$, if $s_i$ is not assigned to process $t_j$, 
it has to transmit $t_j$ to another server and obtain some intermediary fee, so the reward is $\lambda p_j \mu_{i,j}-E_{i,j}^\text{comm}$; but if $s_i$ is able to process $t_j$, it can get the payment $p_i \mu_{i,j}$ at the cost of the computing energy consumption $E_{i,j}^\text{comp}$, and thus, the reward is $U_{i,j}=p_i \mu_{i,j}-E_{i,j}^\text{comp}$.

If $t_j\notin T_{s_i}$, the reward function is as below:
\begin{align}
U_{i,j}=\begin{cases}
(1-\lambda)p_j \mu_{i,j}-E_{i,j}^\text{comp},&\mu_{i,j}\leq\mu_i\ \text{and}\ T_{i,j} \leq \tau_{e,j},\\
0,& \text{otherwise}.
\end{cases}
\label{reward_2}
\end{align}

The logic of (\ref{reward_2}) is similar to (\ref{reward_1}): when $t_j$ is not from $s_i$, if $s_i$ has the capability to process $t_j$, it can get the reward the same as in (\ref{reward_1}) but has to pay the intermediary fee; and it will get nothing if it cannot process this task, which means that $A_i$ as one action cannot be chosen in state $R_j$.

Based on the above analysis, we know that the reward functions are the transformation of the objective function under certain constraints. For simplicity, we summarize the calculation of reward functions in Algorithm~\ref{al_11}. To get an $n \times m$ matrix $\mathcal{U}=\left\{U_{1,1},U_{1,2},\cdots,U_{i,j},\cdots,U_{n,m}\right\}$ containing the reward value for each state and each action,  we need to calculate the energy consumption and time cost for both computing and communication processes (Lines 3-6). And then we can calculate and return the rewards based on (\ref{reward_1}) and (\ref{reward_2}) (Lines 7-25).

\begin{algorithm}
\caption{StateActionReward} 
\label{al_1}
\begin{algorithmic}[1]
\REQUIRE $\Omega$, $\mathcal{A}$
\ENSURE $\mathcal{U}$
\FOR{$i\in \{1,\cdots,n\}$}
\FOR{$j \in \{1,\cdots,m\}$}
\STATE $E_{i,j}^\text{comp}$ $\leftarrow$ $\alpha_i \mu_{i,j} f_{i}^2$
\STATE $T_{i,j}^\text{comp}$ $\leftarrow$ $\frac{\mu_{i,j}}{f_{i}}$
\STATE $T_{i,j}^\text{comm}$ $\leftarrow$ $\frac{D_j}{r_{i}}$
\STATE $E_{i,j}^\text{comm}$ $\leftarrow$ $H_{i} T_{i,j}^\text{comm}$
\IF{$s_j \in T_{s_i}$}
\STATE $T_{i,j} \leftarrow T_{i,j}^\text{comp}\mathds{1}_{i,j}+T_{i,j}^\text{comm} (1-\mathds{1}_{i,j})$
\IF{($\mu_{i,j}\leq \mu_i$) and ($T_{i,j}$ $\leq$ $\tau_{e,j}$)}
\STATE $U_{i,j} \leftarrow p_j \mu_{i,j}-E_{i,j}^\text{comp}$
\ELSE
\STATE$U_{i,j} \leftarrow \lambda p_j \mu_{i,j}-E_{i,j}^\text{comm}$
\ENDIF
\ENDIF
\IF {$s_j \notin T_{s_i}$}
\STATE $T_{i,j}=(T_{i,j}^\text{comp}+T_{i^*,j}^\text{comm})\mathds{1}_{i,j}$
\IF{($\mu_{i,j}\leq \mu_i$) and ($T_{i,j}$ $\leq$ $\tau_{e,j}$)}
\STATE $U_{i,j}\leftarrow (1-\lambda)p_j \mu_{i,j}-E_{i,j}^\text{comp}$
\ELSE
\STATE$U_{i,j} \leftarrow 0$
\ENDIF
\ENDIF
\ENDFOR
\ENDFOR
\RETURN $\mathcal{U}$
\end{algorithmic}
\label{al_11}
\end{algorithm}

\subsection{Learning Process}
\label{learning process}

\subsubsection{Available Actions}
To avoid selecting unmatched servers that cannot process such tasks as actions, 
and to improve the learning efficiency by reducing the time complexity, we need to clarify which set of actions are available in each state before selecting one as the action. Thus, we define the concept of available actions as below.

\begin{definition}[Available Actions] 
\label{ava_actions}
Assume the agent is at state $\Omega_j$, $\mathcal{A}_{j}^{ava}=\left\{s_1,s_2,\cdots,s_{n'}\right\}$ is the set of available actions, i.e., servers, that can process $t_j$ under the constraints of $\mathrm{C1}-\mathrm{C4}$, with $n'\leq n$ being the number of available actions. 
\end{definition}

Since the transfer of states is a dynamic process and the states will affect each other, one of the most direct effects is that once an action is selected, its computing power will be reduced and therefore will constrain the action selection of the subsequent states. Therefore, we evaluate the available actions of each state one by one. 
Recall the discussions in Section~\ref{reward_function}, we get an $n\times m$ rewards table $\mathcal{U}$, and we can use $\mathcal{U}_j=\left\{U_{1,j},U_{2,j},\cdots,U_{n,j}\right\}$ to denote all the possible rewards of $\Omega_j$.
Generally speaking, $\mathcal{U}_j$ has three types of value: positive, zero, and negative. We can select those actions with non-zero values as the available actions. However, this naive method can lead to some severe consequences. For example, one server is selected multiple times due to its non-zero value (this could happen in Q-learning, especially when the number of servers is small), but its capacity of available computing resources is not enough to handle all of these tasks even if the reward for each task is positive or optimal.

To address this challenge, we design a dynamic table 
to record the accumulative $\mu_{i,j}$, which can be denoted by $\mu_{i,j}^{acc}$. This table shares the same dimensions with the rewards table $\mathcal{U}$ generated in Algorithm~\ref{al_11}. Hence, we can design the following mechanism to get the available actions for each state.

First, naively select the actions with non-zero values from $\mathcal{U}$ and initialize the values as zero to reduce the computational cost by avoiding going through the whole action space again. Let $\mathcal{A}_j^{nz}=\left\{A_{1,j}^{nz},A_{2,j}^{nz},\cdots,A_{n',j}^{nz}\right\}$ be the set to contain the actions with non-zero values at state $\Omega_j$. In other words, the agent only needs  to check the actions in $\mathcal{A}_j^{nz}$ at state $\Omega_j$.
Then, based on the selected actions, we can get a value of $\mu_{i,j}$. In this way, after multiple rounds, we can get the accumulative value $\mu_{i,j}^{acc}$. This is a dynamic process, which means that only those selected actions can contribute to their corresponding accumulative values. Besides, we need to consider the time constraint, ensuring that one server can compute the allocated tasks in time. We use $\tau_{e,i}^{ava}$ to demonstrate the available computing time for the actions in $\mathcal{A}_{j}^{nz}$, which is calculated via $\tau_{e,i}^{ava}=(\mu_{i}-\mu_{i,j}^{acc})/f_{i}$.
Next, we can decide whether an action ${A}_{i,j}^{nz}$ is available by the following rules: if $\mu_{i,j}^{acc}<\mu_i$ and $\tau_{e,j}^{ava}>\tau_{e,j}$, then $\mathcal{A}_{i,j}^{nz}$ is considered as available action; otherwise, ${A}_{i,j}^{nz}$ will be discarded from $\mathcal{A}_j^{nz}$. In this way, we can get a set of available actions $\mathcal{A}_j^{ava}$ at state $\Omega_j$, and the final action ${A}_j$ should be chosen from it according to the selection policy which will be discussed in Section~\ref{action_selection}.

\begin{algorithm}[b!]
\caption{StateAvailableActions} 
\begin{algorithmic}[1]
\REQUIRE $\mathcal{A}$, $\Omega_j$, $\mathcal{U}$
\ENSURE ${A}_j$
\STATE Initialize $\mathcal{A}$, $\mathcal{A}_j^{ava}$, $\mu_{i,j}^{acc}$
\FOR {$i \in \{1,\cdots,n\}$}
\IF {$U_{i,j} \neq 0$}
\STATE Append $A_i$ into $\mathcal{A}_j^{nz}$
\ENDIF
\ENDFOR
\STATE Choose $A_{i,j}^{nz}$ from $\mathcal{A}_j^{nz}$
\STATE $\tau_{e,i}^{ava}\leftarrow (\mu_i-\mu_{i,j}^{acc})/f_{i}$
\IF {$A_{i,j}^{nz} \in \mathcal{A}_j^{ava}$}
\IF {($\mu_{i,j}^{acc} > \mu_i$) or ($t_{i,j}^{ava} < \tau_{e,j}$)}
\STATE Remove $A_{i,j}^{nz}$ from $ \mathcal{A}_j^{ava}$
\STATE $A_j$ $\leftarrow$ choose another action from $ \mathcal{A}_j^{nz}$
\ENDIF
\ELSE
\STATE Append $A_{i,j}^{nz}$ into $\mathcal{A}_j^{ava}$
\STATE ${A}_j \leftarrow A_{i,j}^{nz}$
\ENDIF
\STATE $\mu_{i,j}^{acc}$ $\leftarrow$  $\mu_{i,j}^{acc}+\mu_{i,j}$
\RETURN ${A}_j$
\end{algorithmic}
\label{al_2}
\end{algorithm}

The whole process of choosing available actions for a specific state is summarized in Algorithm 2. At the beginning, we initialize $\mathcal{A}_j^{nz}$, $\mathcal{A}_j^{ava}$, and $\mu_{i,j}^{acc}$ (Line 1), and we get the non-zero values from $\mathcal{U}$ (Lines 2-6). Next, we choose an action from $\mathcal{A}_j^{nz}$ and calculate its available computing time (Lines 7-8).
Then, we can determine whether the selected action is available or not (Line 9-17), and we update $\mu_{i,j}^{acc}$ (Line 18).
Finally,  we get the available action ${A}_j$ at state $\Omega_j$. Please note that the selection method in Lines 7 and 12 will be further discussed in Section~\ref{action_selection}. In general, Algorithm~\ref{al_2} is the transformation of $\mathrm{C1}$ and $\mathrm{C2}$ in the MTA problem.

\subsubsection{The Update of Q-table}
\label{action_selection}


The Q-table is applied during the learning process, which is used to facilitate the action selection. First, we generate an $m\times n$ matrix as the Q-table $Q(\Omega, \mathcal{A})$, and initialize its value with 0. Then, we need to update the Q-value at each state. Here we use the following equation to update Q-value:
\begin{align*}
    &Q(\Omega_j, A_i)\\
    &=Q(\Omega_j, A_i)+\alpha[U_{i,j}+\gamma \max Q(\Omega'_j, A_j^{ava})-Q(\Omega_j, A_i)],
\end{align*}
where $Q(\Omega'_j, A_j^{ava})$ means all the possible Q-value at next stage; $0\leq \alpha \leq 1$ is the learning rate and $0\leq \gamma \leq 1$ is the discount factor.

\begin{algorithm}[b!]
\caption{Learning-based Algorithm for the MTA Problem} 
\begin{algorithmic}[1]
\REQUIRE $\Omega, \mathcal{A}, K$
\ENSURE $TS^*$
\STATE Initialize $Q(\Omega,\mathcal{A})$, $k$
\STATE $\mathcal{U} \leftarrow$ StateActionReward($\Omega, \mathcal{A}$)
\STATE $ A_{j}^{nz} \leftarrow$ actions with non-zero values in $\mathcal{U}$ 
\STATE Generate a random value $x$
\IF{$x \leq \epsilon$}
\STATE $A_j \leftarrow$ randomly selected from $\mathcal{A}_{j}^{nz}$
\ELSE 
\STATE $A_j$ $\leftarrow$ $\max Q(\Omega_j,.)$
\ENDIF 
\FOR{$k \in \{1,\cdots,K\}$}
\FOR{$j \in \{1,\cdots,m\}$}
\STATE $A_j \leftarrow$ StateAvailableActions($\mathcal{A}$, $\Omega_j$, $\mathcal{U}$)
\STATE $Q(\Omega_j,A_i) \leftarrow Q(\Omega_j,A_i)+\alpha[U_{i,j}+\gamma \max Q(\Omega'_j, \mathcal{A}_j^{ava})-Q(\Omega_j, A_i)]$
\ENDFOR
\STATE $TS \leftarrow $ the sum of all rewards in this episode
\ENDFOR
\STATE $TS^* \leftarrow \max TS$
\RETURN $TS^*$
\end{algorithmic}
\label{al_3}
\end{algorithm}

The agent selects the action based on the Q-value, and we adopt the $\epsilon$-greedy algorithm as the policy of action selection. The $\epsilon$-greedy strategy selects the action with the largest expected rewards most of the time, and the parameter $\epsilon$ balances exploration and exploitation. We can use a larger value of $\epsilon$ to allow the agent to exploit what have been learned, and the agent will explore more actions that have not been learned with a smaller $\epsilon$. 
The learning process requires multiple rounds so that the agent can learn more about the rewards and find the best solution. For each episode $k\in \left\{1,2,\cdots,K\right\}$ where $K$ is the total episodes of learning, we let the agent go through all the states and find a solution, and then we compare all the solutions $TS$ and choose the one $TS^*$ with the highest sum of rewards as the optimal solution.

Algorithm 3 is the detailed process of the learning-based solution for the MTA problem. We first initialize $Q(\Omega,\mathcal{A})$, $k$, and get $U_{i,j}$ via Algorithm~\ref{al_11} (Lines 1-2). The $\epsilon$-greedy strategy is implemented to select the action (Lines 3-9). Then, the learning process will last until episode $k$ reaches the predefined total number of episodes $K$ (Lines 10-16), and the optimal solution will be obtained and returned (Line 17-18).

\subsection{Complexity Analysis}

The learning-based algorithm comprises three sub-algorithms, so the complexity analysis needs to take all of them into account. First, the computational complexity of Algorithm~\ref{al_11} is $O(m\times n)$, and the complexity of Algorithm~\ref{al_2} is $O(n)$. As for the time complexity of Algorithm~\ref{al_3}, it is $O(k\times n)$. Thus, in general, the time complexity of the learning-based solution should be $O(m\times n)=O(k\times n)+O(n)+O(m\times n)$, which means that we can solve the MTA problem in polynomial time.

\section{Experimental Evaluation}
\label{experiments}


\subsection{Experimental Setting}

In our experiments\footnote{The code is available in: {https://github.com/wzljerry/Blockchain-based-Edge-Resource-Sharing-for-Metaverse}}, we consider a blockchain-based MEC system in Metaverse with 20 servers and 50 offloading tasks as the primary setting. We change relevant parameters to analyze the impacts of the numbers of servers and tasks on the total rewards and the time complexity, as well as the influence of Q-learning parameters on the convergence. For clarity, we summarize the basic parameter settings in Table~\ref{tb_1}.

\begin{table}[h!t]
\centering
\caption{Basic Parameter Settings}
\begin{tabular}{|l|l|l|l|}
\hline
$n=20$                       & $m=50$                       & $D_i=[200,400]$       & $\theta=0.01 $   \\ \hline
$p_j=[1,10]$                & $D_j=[10,20]$               & $\tau_{e,j}=[1,100]$    &    $\gamma=0.9$             \\ \hline
$f_i=[1,10]$                & $H_i=[5,10]$                & $G_i=[5,10]$      &$\alpha=0.01$            \\ \hline
$\delta=0.01$ & $B_i=[5,10]$                & $\epsilon=0.9$  & $K=500$  \\ \hline
\end{tabular}
\label{tb_1}
\end{table}


\subsection{Experimental Results}

First, we design experiments to prove the effectiveness of our proposed learning-based algorithm for the MTA problem. 
To that aim, we compare the learning-based solution with two benchmark algorithms, i.e., the random allocation algorithm and the greedy-based algorithm. Specifically, The random allocation algorithm allocates each task to one server randomly; and the greedy-based algorithm selects the server with the maximum reward for each task.
We run these three algorithms with different numbers of servers and tasks to obtain  
the comparison results as shown in Fig.~\ref{fig_1}. From Figs.~\ref{fig_1}(a) and~\ref{fig_1}(b), we can see that the learning-based algorithm can always outperform the other two algorithms with higher rewards. The random strategy fluctuates a lot due to the randomness of each selection. The greedy and learning algorithms, in contrast, perform more smoothly. Intuitively, more servers will not result in higher total rewards; however, more tasks would increase the overall revenue. But from our experimental results, this trend does not hold. There are at least three reasons for it. First, C1--C4 constrain the allocation decisions, resulting in dynamic decisions as the number of servers or tasks changes, which does not have a cumulative effect on rewards; second, although more tasks result in more rewards, they also increase the cost of computation and communication; and third, servers and tasks are heterogeneous with different parameters, exacerbating the results from the first two reasons. 

\begin{figure}[h]
\centering
\subfigure[The Number of the MEC Servers]{
\includegraphics[width=0.465\linewidth]{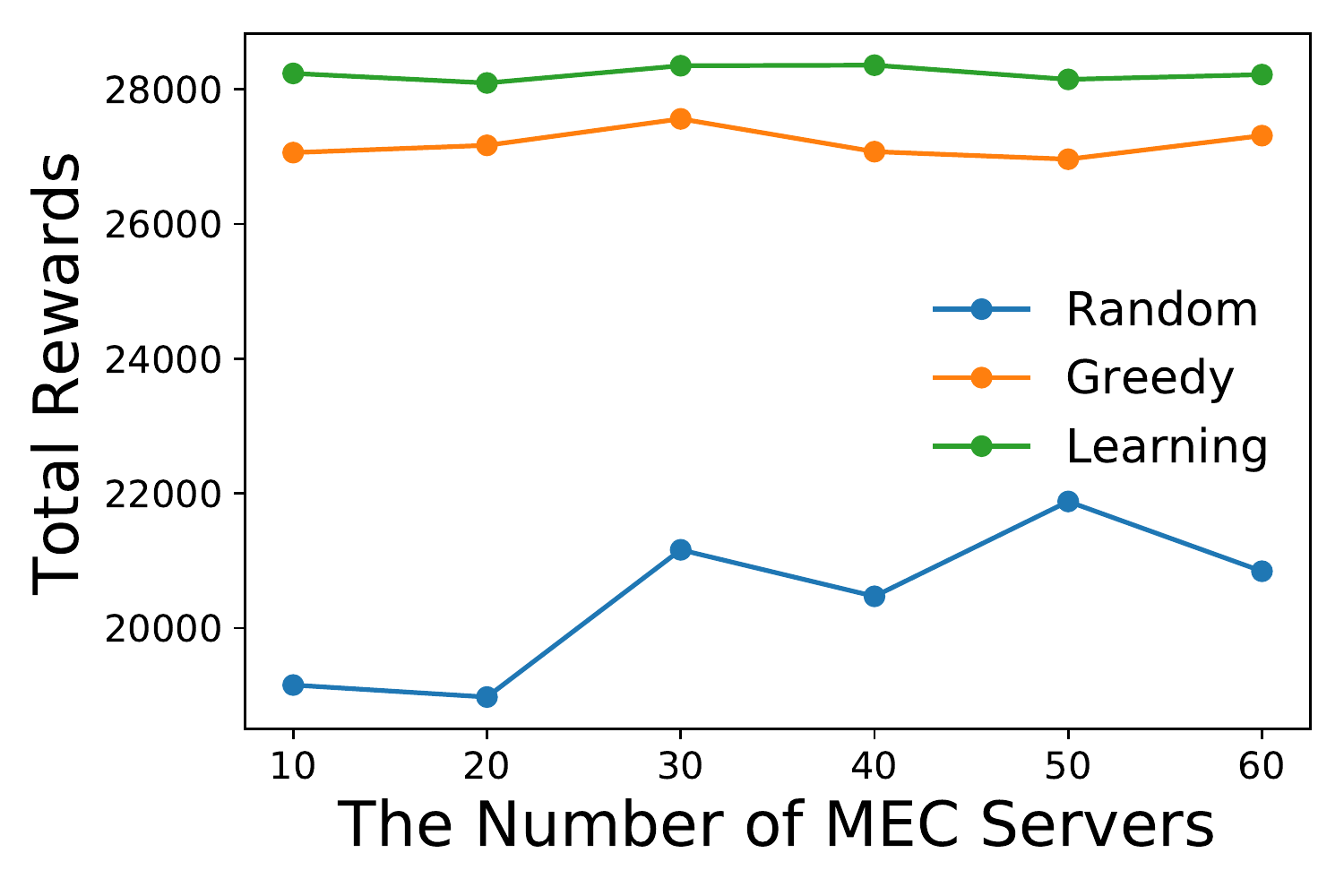}}\hfill
\subfigure[The Number of Offloading Tasks]{
\includegraphics[width=0.465\linewidth]{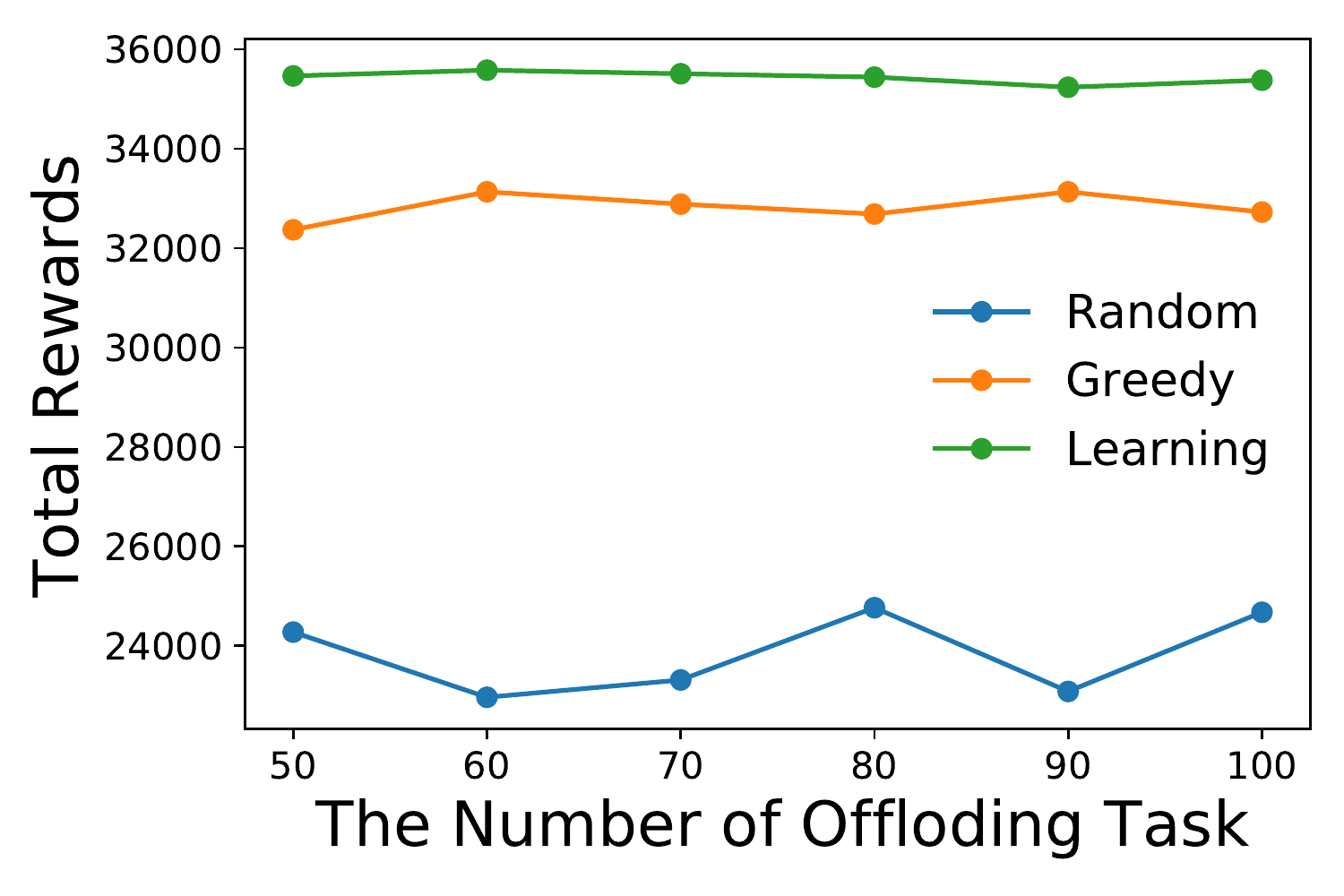}}
\caption{The influences of the numbers of MEC servers and offloading tasks on the total rewards.}
\label{fig_1}
\end{figure}

Then, we explore the influences of unit prices and data sizes of offloading tasks on the total rewards. We increase $p_j$ and $D_j$ with the percentage from 10\% to 100\%, and the results are shown in Fig.~\ref{fig_5}. It is clear that both the increase of $p_j$ and $D_j$ will affect the total rewards.  When the unit prices are larger, the users will pay more to the system, so the total rewards will increase. And, if the data size of $t_j$ is larger, then the CPU cycles required to process $t_j$ will also be increased, and thus the total rewards  will be higher.

\begin{figure}[h]
\centering
\subfigure[The Increase Percentage of Unit Price]{
\includegraphics[width=0.465\linewidth]{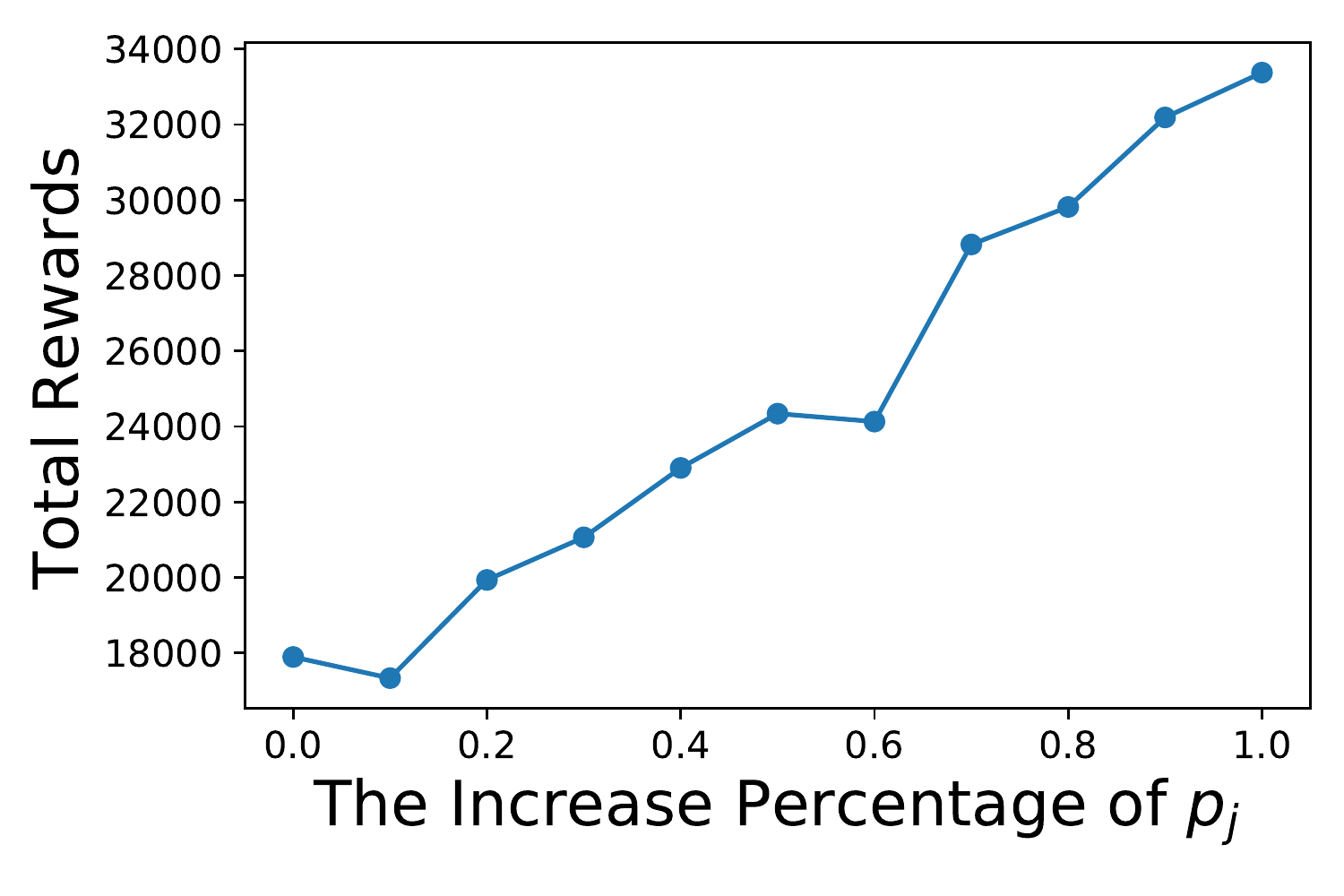}}\hfill
\subfigure[The Increase Percentage of Data Size]{
\includegraphics[width=0.465\linewidth]{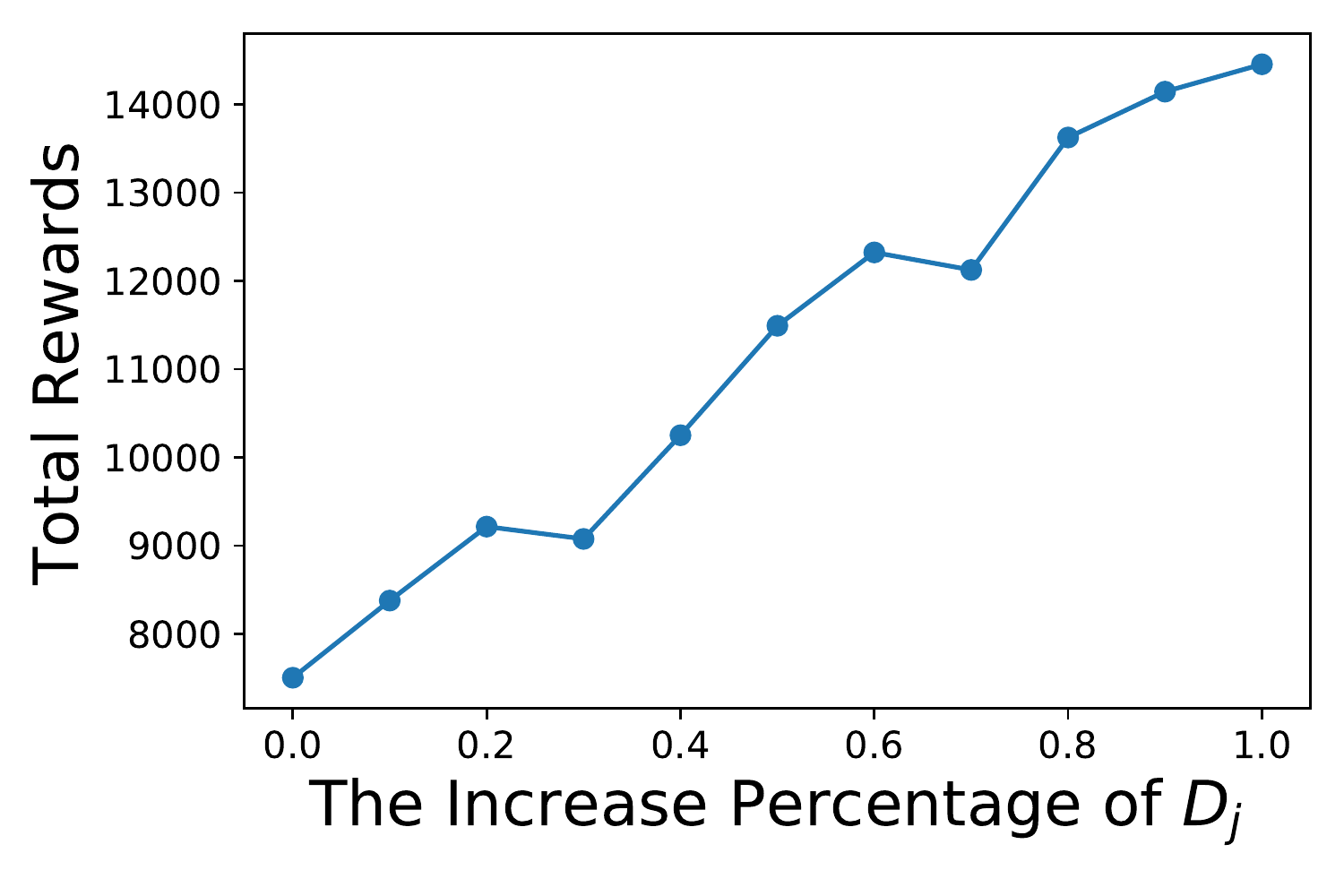}}
\caption{The influences of unit prices and data sizes of offloading tasks on the total rewards.}
\label{fig_5}
\end{figure}



Last but not least, we mainly examine the convergence performance of the learning algorithm. Q-learning is determined by two main variables, i.e., $\alpha$ and $\gamma$. 
From Fig.~\ref{fig_2}(a), we can see that when $\alpha$ is smaller, the current choice is more influenced by experience and lacks further exploration, so although it achieves a high payoff at the beginning, the convergence rate does not improve effectively. 
From Fig.~\ref{fig_2}(b), we can see that the larger the value of $\gamma$, the better the convergence. This is because $\alpha$ takes into account the effect of future rewards on current choices, so as $\gamma$ becomes larger, it emphasizes rewards more and thus learns the optimal combinations more efficiently to speed up the convergence. 

\begin{figure}[h!t]
\centering
\subfigure[The Learning Rate $\alpha$.]{
\includegraphics[width=0.465\linewidth]{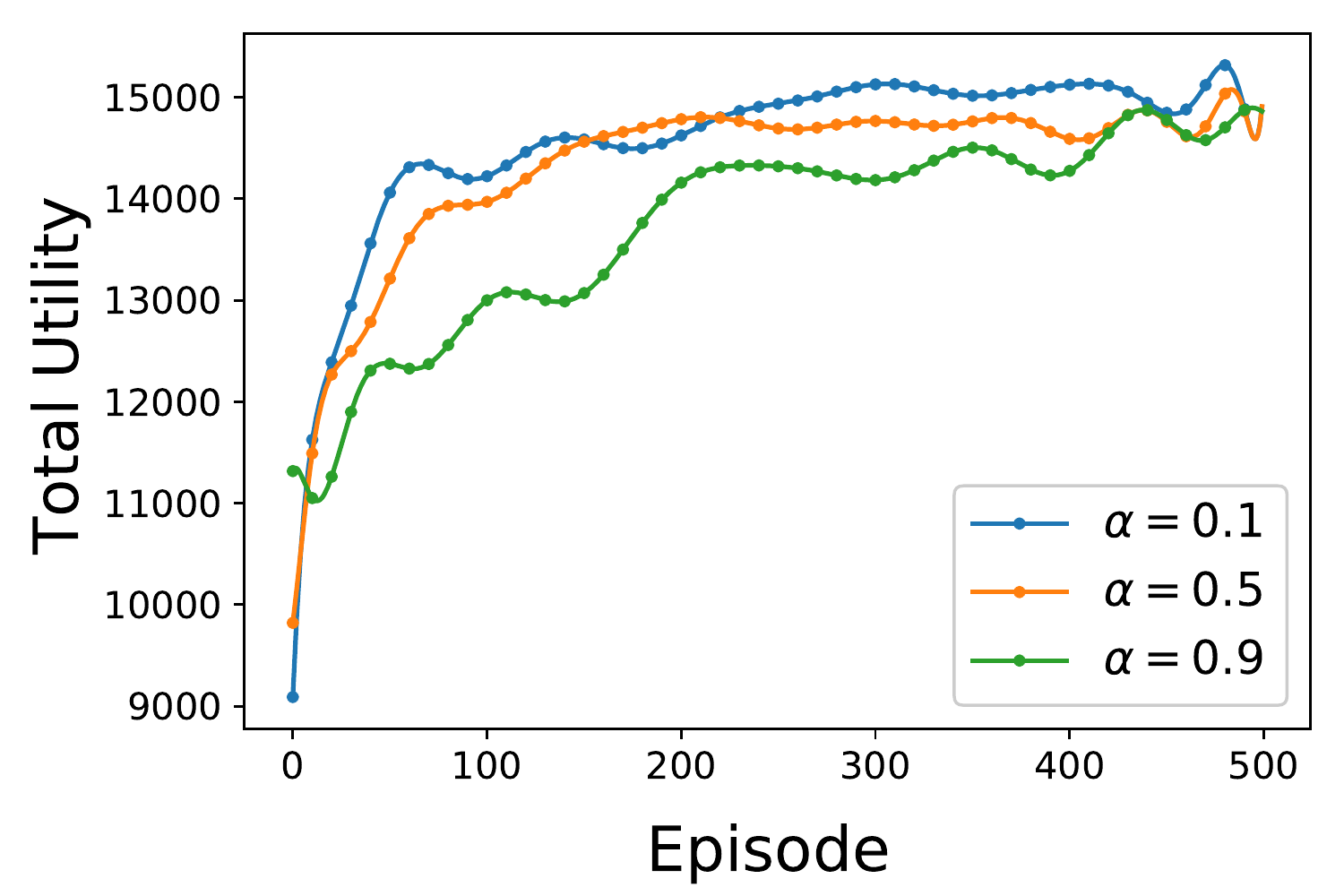}}\hfill 
\subfigure[The Discount Factor $\gamma$.]{
\includegraphics[width=0.465\linewidth]{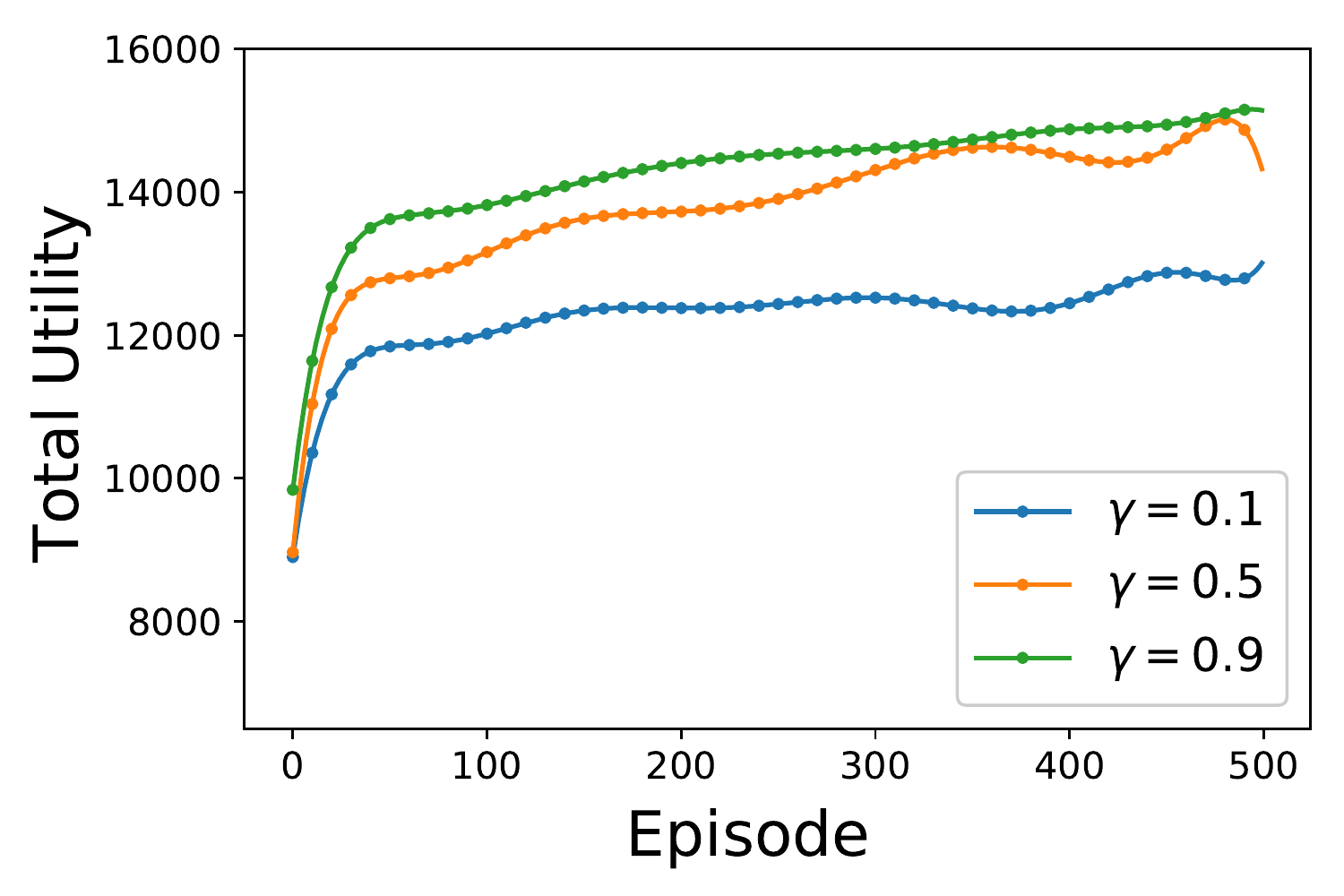}}
\caption{The convergence of the learning-based algorithm.}
\label{fig_2}
\end{figure}

\section{Conclusion}
\label{conclusion}
In this paper, we establish a blockchain-based MEC platform for resource sharing and optimization to facilitate Metaverse applications. In particular, we design a task allocation scheme to assist our proposed system. To that aim,  a learning-based algorithm is proposed to help the system make task allocation decisions in polynomial time. Numerous experiments prove that our proposed system and algorithms are efficient. 


\bibliographystyle{IEEEtran}
\bibliography{reference.bib}

\end{document}